\def\anon{0}
\def\draft{1}

\documentclass[11pt]{article}
\usepackage{color,colortbl,latexsym,amsthm,amsmath,amssymb,path,enumitem,soul,tikz,fullpage,subcaption}
\usetikzlibrary{shapes.geometric}
\usepackage[colorlinks,citecolor=blue,linkcolor=blue,urlcolor=red,pagebackref]{hyperref}
\usepackage{algorithm}
\usepackage{algorithmicx}
\usepackage{algpseudocode}
\usepackage{soul}
\usepackage{multirow}
\usepackage{thmtools}

\usepackage{relsize}
\usepackage{diagbox}
\usepackage{comment}
\usepackage[font={small,sf}]{caption}
\newcommand{\ignore}[1]{}
\makeatother


\setlist[itemize]{itemsep=-0.1em}
\setlist[enumerate]{itemsep=-0.1em}

\makeatletter
\def\th@plain{%
  \thm@notefont{}% same as heading font
  \itshape % body font
}
\def\th@definition{%
  \thm@notefont{}% same as heading font
  \normalfont % body font
}
\makeatother
\newtheorem{theorem}{Theorem}[section]
\newtheorem{lemma}[theorem]{Lemma}

\newtheorem{hypothesis}[theorem]{Hypothesis}
\newtheorem{question}{Question}
\newtheorem{claim}[theorem]{Claim}
\newtheorem{definition}[theorem]{Definition}

\usepackage{graphicx,psfrag}

\usepackage{listings,xcolor}
\usepackage{dsfont}

\newcommand{\FF}{\ensuremath{\mathbb F}}

\renewcommand{\epsilon}{\varepsilon}

\ifnum\draft=1
\newcommand{\surya}[1]{\textcolor{magenta}{(Surya: #1)}}

\newcommand{\virgi}[1]{\textcolor{red}{(Virginia: #1)}}
\else
    \newcommand{\mina}[1]{}
    \newcommand{\surya}[1]{}
    \newcommand{\yinzhan}[1]{}
    \newcommand{\virgi}[1]{}
\fi 
\newcommand\cliquelist[1]{\ensuremath{#1\text{-}\mathsf{Clique}\text{-}\mathsf{Listing}}}

\newcommand{\exactclique}[1]{\ensuremath{\mathsf{Exact}\text{-}#1\text{-}\mathsf{Clique}}}

\newcommand\fourcycle{\ensuremath{4\text{-}\mathsf{Cycle}\text{-}\mathsf{Listing}}}
\newcommand\trianglelist{\ensuremath{\mathsf{Triangle}\text{-}\mathsf{Listing}}}
\newcommand{\threesum}{\ensuremath{\mathsf{3SUM}}}
\newcommand{\apsp}{\ensuremath{\mathsf{APSP}}}

\title{A Note on the Conditional Optimality of\\
Chiba and Nishizeki's Algorithms}

\ifnum\anon=1
    \author{Anonymous}
\else
    \author{Yael Kirkpatrick\thanks{The first author was supported in part by the National Science Foundation Graduate Research Fellowship under Grant No 2141064.} \and Surya Mathialagan\thanks{The second author was supported in part by DARPA under Agreement Number HR00112020023, NSF CNS-2154149, a Simons Investigator award, a Thornton Family Faculty Research Innovation Fellowship from MIT and a Simons Investigator Award and by Jane Street. Any opinions, findings and conclusions or recommendations expressed in this material are those of the author(s) and do not necessarily reflect the views of the United States Government or DARPA.}}
\fi

\begin{document}

\maketitle

\begin{abstract}
    In a seminal work, Chiba and Nishizeki [SIAM J. Comput. `85] developed subgraph listing algorithms for triangles, 4-cycle and $k$-cliques, where $k \geq 3.$ The runtimes of their algorithms are parameterized by the number of edges $m$ and the arboricity $\alpha$ of a graph. The arboricity $\alpha$ of a graph is the minimum number of spanning forests required to cover it. Their work introduces:
    \begin{itemize}
        \item A triangle listing algorithm that runs in $O(m\alpha)$ time.
        \item An output-sensitive $\mathsf{4}\text{-}\mathsf{Cycle}\text{-}\mathsf{Listing}$ algorithm that lists all 4-cycles in $O(m\alpha + t)$ time, where $t$ is the number of 4-cycles in the graph. 
        \item A $k\text{-}\mathsf{Clique}\text{-}\mathsf{Listing}$ algorithm that runs in $O(m\alpha^{k-2})$ time, for $k \geq 4.$ 
    \end{itemize}

    Despite the widespread use of these algorithms in practice, no improvements have been made over them in the past few decades. Therefore, recent work has gone into studying lower bounds for subgraph listing problems. 
    The works of 
    Kopelowitz, Pettie and Porat [SODA `16] and Vassilevska W. and Xu [FOCS `20] showed that the triangle-listing algorithm of Chiba and Nishizeki is optimal under the $\mathsf{3SUM}$ and $\mathsf{APSP}$ hypotheses respectively. However, it remained open whether the remaining algorithms were optimal.

    In this note, we show that in fact all the above algorithms are optimal under popular hardness conjectures. First, we show that the $\mathsf{4}\text{-}\mathsf{Cycle}\text{-}\mathsf{Listing}$ algorithm is tight under the $\mathsf{3SUM}$ hypothesis following the techniques of Jin and Xu [STOC `23], and Abboud, Bringmann and Fishcher [STOC `23] . Additionally, we show that the $k\text{-}\mathsf{Clique}\text{-}\mathsf{Listing}$ algorithm is essentially tight under the exact $k$-clique hypothesis by following the techniques of Dalirooyfard, Mathialagan, Vassilevska W. and Xu [STOC `24].
    These hardness results hold even when the number of 4-cycles or $k$-cliques in the graph is small.
\end{abstract}

\section{Introduction}
% \surya{ACTUAL INTRODUCTION}

% \begin{itemize}
%     \item Finding subgraphs important.
%     \item Chiba Nishizeki did this.
%     \item Cite theorems.
%     \item Other works also do this, but Chiba Nishizeki remains the most commonly used in practice due to BLA BLA BLA, especially when the arboricity 
%     \item For eg. because real world graphs have low arboricity.
%     \item However, there have not been any significant improvements on the runtime.
%     \item Lower bounds for triangle listing
%     \item While it is possible that a graph with $m$ edges and arboricity $\alpha$ has $\Theta(m\alpha^{k-2})$ $k$-cliques, is it possible to beat this runtime in graphs with fewer $k$-cliques?
%     \item Similarly, is the $O(m\alpha)$ pre-processing time for 4-cycle listing optimal?
%     \item state result theorems
%     \item related work
% \end{itemize}

Graph pattern listing is a fundamental problem in graph algorithms with various applications. For example, listing $k$-cliques is a problem with a long history of study due to its wide variety of applications in practice  (see \cite{danisch2018listing}, and the many citations within). In social network analysis, one is often interested in listing cliques in order to understand the degree of clustering in graphs \cite{dourisboure2009extraction}. In bioinformatics, one is often interested in finding cliques to identify dense subgraphs in protein-protein interactions and biological networks \cite{fratkin2006motifcut}. Finding $k$-cycles is also valuable in many applications, such as fraud detection \cite{qiu2018real}.
Therefore, many theoretical and practical works have studied the problem of subgraph detection and enumeration \cite{danisch2018listing,listingcliquesdensest,listingcliquesnucleus,count5via3,trilistlatapy,SchankW05,ShunT15,ChuC11,bjorklund2014listing, jin2023removing, dalirrooyfard2023listing, abboud2022listing}.

In the seminal work of Chiba and Nishizeki \cite{chiba1985arboricity}, they introduced an edge-searching technique which allows for various pattern enumeration algorithms parameterized by graph arboricity\footnote{The arboricity $\alpha$ of a graph is the minimum number of spanning forests required to cover it.}. Their work provides three main algorithms for enumerating triangles, 4-cycles, and $k$-cliques. 

\begin{theorem}[Triangle enumeration]\label{thm:triangle-list}
    Given a graph $G$ with $m$ edges and arboricity $\alpha$, there exists an algorithms that enumerates all triangles in $O(m\alpha)$ time. 
\end{theorem}

Chiba and Nishizeki also provide an algorithm for listing 4-cycles which first preprocesses a data structure in time $O(m\alpha)$ and then outputs each 4-cycle with $O(1)$ delay. 
\begin{theorem}[4-cycle enumeration]\label{thm:CN-four-cycle}
    Given a graph $G$ with $m$ edges, arboricity $\alpha$ and $t$ 4-cycles, there exists an algorithm that lists all 4-cycles in $O(m\alpha + t)$ time.
\end{theorem}

Finally, Chiba and Nishizeki extend their result for triangle enumeration to $k$-clique enumeration, for any $k\geq 3$.
\begin{theorem}[$k$-clique enumeration]
    Assume $k \geq 3$. Given a graph $G$ with $m$ edges and arboricity $\alpha$, there exists an algorithms that enumerates all $k$-cliques in $O(m\alpha^{k-2})$ time. 
\end{theorem}

All the above algorithms do not require their user to compute the arboricity of the graph, but rather benefit from the fact that in many applications, input graphs implicitly have low arboricity \cite{shi2021parallel}. For instance, planar graphs have constant arboricity, resulting in a linear time triangle and 4-cycle listing algorithms. Because of this, in addition to the simplicity of these algorithms, all the above algorithms are frequently used in practice. Alon, Yuster and Zwick \cite{colorcoding} use Chiba and Nishizeki's triangle listing algorithm when dealing with planar graphs. Yin et al. \cite{localclustering} use Chiba and Nishizeki's $k$-clique listing algorithm as part of new algorithms for local graph clustering. These techniques are further used by Alon, Yuster and Zwick \cite{findingcycles}, where they extend the ideas of Chiba and Nishizeki's 4-cycle enumeration algorithm to find and count cycles of various lengths.

Despite their widespread use in practice, little improvement has been made in the past few decades over the runtime of these algorithms. Some polylogarithmic improvements have been achieved over the $O(m\alpha)$ runtime for triangle listing \cite{kopelowitz2015dynamic, eppstein2017}, but no polynomial improvements are known for any of the algorithms. Therefore, various works have recently studied conditional lower bounds for triangle listing. In fact, the work of Kopelowitz, Pettie and Porat \cite{kopelowitz2016higher} showed that \autoref{thm:triangle-list} is conditionally optimal under \threesum{}, and Vassilevska W. and Xu~\cite{williams2020monochromatic} showed it is optimal under \apsp{}.

However, no similar lower bound is known for the remaining algorithms. This leads us to the following questions. 

    \begin{question}\label{qn:4-cycle}
        Is there a $4$-cycle listing algorithm with $O(m^{1-\epsilon}\alpha)$ pre-processing time, and $m^{o(1)}$ delay to list each $4$-cycle, for a constant $\epsilon > 0$?
    \end{question}

In the case of the $k$-clique listing algorithms for $k \geq 3$, it is possible that a graph $G$ with $m$ edges and arboricity $\alpha$ indeed has $\Theta(m\alpha^{k-2})$ $k$-cliques. For the case of $k = 3$, the results of \cite{kopelowitz2016higher}
and \cite{williams2020monochromatic} shows that $O(m^{1-o(1)}\alpha)$ time is necessary even when the graph has $o(m^{1-o(1)} \alpha)$ triangles. Naturally, we can ask if a similar result holds for $k \geq 4$. 

\begin{question}\label{qn:k-clique}
    Is there a $k$-clique listing algorithm with $O(m^{1-\epsilon} \alpha^{k-2})$ runtime on graphs with $o(m^{1-\epsilon} \alpha^{k-2})$ $k$-cliques, for some $\epsilon > 0$?
\end{question}

\subsection{Our Contributions}
In this work, we answer both Questions \ref{qn:4-cycle} and \ref{qn:k-clique} and show that all the algorithms presented by Chiba and Nishizeki \cite{chiba1985arboricity} are in fact essentially optimal. 

% To prove the 4-cycle enumeration algorithm is optimal, we extend the techniques of \cite{abboud2023stronger, jin2023removing}. These works concurrently showed that under the \threesum{} hypothesis, a $\Omega(m^{4/3 - o(1)})$ pre-processing time is in fact necessary for the 4-cycle listing with $m^{o(1)}$ delay. However, Chiba and Nishizeki's algorithm in \autoref{thm:CN-four-cycle} beats this runtime when $\alpha = o(m^{1/3})$. Therefore, we adapt their techniques to show a matching lower bound on graphs with small arboricity. 
First, we answer Question \ref{qn:4-cycle} in the affirmative and show that the 4-cycle enumeration algorithm is optimal. In previous work, Abboud, Bringmann and Fischer \cite{abboud2023stronger} and Jin and Xu \cite{jin2023removing}, showed that under the \threesum{} hypothesis, a $\Omega(m^{4/3 - o(1)})$ pre-processing time is in fact necessary for the 4-cycle listing problem with $m^{o(1)}$ delay. However, Chiba and Nishizeki's algorithm in \autoref{thm:CN-four-cycle} beats this runtime when $\alpha = o(m^{1/3})$. We extend their result and prove the following theorem.
% To do so we extend the techniques of \cite{abboud2023stronger, jin2023removing} to prove that Chiba and Nishizeki's 4-cycle enumeration algorithm is tight, and extend the techniques of \cite{dalirrooyfard2023listing} to show the $k$-cycle enumeration algorithm is essentially tight. More precisely, we prove the following two theorems.
% \surya{To do this, we extend the techniques of \cite{abboud2023stronger, jin2023removing} and \cite{dalirrooyfard2023listing} to obtain our results.}

\begin{restatable}{theorem}{FourCycleLB}\label{thm:four_cycle_lb}
    Under the \threesum{} hypothesis, there is no algorithm with $O(m^{1-\epsilon}\alpha)$ pre-processing time and $m^{o(1)}$ delay that solves 4-cycle enumeration on graphs with arboricity $\alpha = O(m^{1/3-\mu})$, for any constants $\epsilon,\mu > 0.$%$\Omega(m\alpha^{1-o(1)})$ time is required to list all $4$-cycles in a graph with $m$ edges and arboricity $\alpha$. 
\end{restatable}

Next, we show that the $k$-clique listing algorithm of Chiba and Nishizeki is essentially optimal. We derive our hardness results from the \emph{Exact-$k$-Clique hypothesis}. 
% To show a lower bound for $k$-clique listing, we refine the techniques of Dalirrooyfard et al. \cite{dalirrooyfard2023listing}, who showed hardness for output-sensitive $k$-clique listing from the \emph{Exact-$k$-Clique hypothesis}. 

\begin{hypothesis}[Exact-$k$-Clique hypothesis]\label{hyp:exact_k_clique}
For a constant $k\geq 3$, let $\exactclique{k}$ be the problem that given an $n$-node graph with edge weights in $\{-n^{100k},\dots, n^{100k}\}$, asks to determine whether the graph contains a $k$-clique whose edges sum to $0$.
Then, $\exactclique{k}$ requires $n^{k-o(1)}$ time, on the word-RAM model of computation with $O(\log n)$ bit words. 
\end{hypothesis}

The Exact-$k$-Clique hypothesis is among the popular hardness hypotheses in fine-grained complexity. It has been used to prove hardness for output-sensitive $k$-clique listing~\cite{dalirrooyfard2023listing}, the Orthogonal Vectors problem in moderate dimensions \cite{abboud2018more} and join queries in databases \cite{BringmannCM22}. Moreover, due to known reductions  (see e.g. \cite{vsurvey}), the Exact-$k$-Clique hypothesis is at least as believable as  the Max-Weight-$k$-Clique hypothesis which is used in many previous papers (e.g. \cite{AbboudWW14,BackursDT16,BackursT17,LincolnWW18,BringmannGMW20}). 

In this work, we follow the techniques of Dalirrooyfard, Mathialagan, Vassilevska W. and Xu~\cite{dalirrooyfard2023listing} and reduce \exactclique{k} to the problem of $k$-clique listing in graphs of bounded degree, to prove the following theorem.

\begin{restatable}{theorem}{KCliqueLB}\label{thm:k_clique_lb}
    Let $\alpha = n^{1-\delta}$ be the arboricity of a given graph. Under the \exactclique{k} hypothesis, $\Omega(m^{1-O(\delta)}\alpha^{k-2})$ time is required to list all $k$-cliques in the graph. This is true even when the graph has $o(m^{1-O(\delta)}\alpha^{k-2})$ $k$-cliques.
\end{restatable}

\section{Preliminaries}

\subsection{Graph Theory Tools and Definitions}
Throughout this paper, we focus on a graph parameter called graph \emph{arboricity}, which intuitively captures the density of a graph. Formally, the arboricity of a graph is defined as follows.
\begin{definition}
    Given a graph $G$, define its \emph{arboricity} $\alpha(G)$ to be the minimum number of edge-disjoint spanning forests into which $G$ can be decomposed. 
\end{definition}

When $G$ is clear from context we denote its arboricity simply by $\alpha$. In our theorems we often use degree bounded graphs, in which case the following claim is useful to relate the graph arboricity to the maximum degree.

\begin{claim}
    $\alpha(G) \leq \max_v \deg(v)$.
\end{claim}

The proof of this claim follows from noting that when deconstructing a graph into spanning forest, every tree reduces the maximum degree of the graph by at least one.

\subsection{Main Problems}
Next, we define the main subgraph enumeration problems we study in this paper.
\begin{definition}[\fourcycle{}]
    Given a graph $G = (V, E)$, list all $4$-cycles in the graph.
\end{definition}

\begin{definition}[\cliquelist{k}]
    Given a graph $G = (V, E)$, list all $k$-cliques in the graph. 
\end{definition}

\subsection{Main Hypotheses}
Finally, we define the hypotheses on which we base our reductions in this paper.

\begin{definition}[\threesum{}]
    Given $n$ integers in the range $\{-n^c, \dots, n^c\}$ for constant $c$, determine if there exist three integers that sum to zero.
\end{definition}

\begin{hypothesis}[\threesum{} hypothesis]
    In the word RAM model with $O(\log n)$ bit words, there is no $O(n^{2-\epsilon})$ algorithm for \threesum{}, for any $\epsilon > 0$.
\end{hypothesis}

\begin{definition}[\exactclique{k}]
     For a constant $k\geq 3$, given an $n$-node graph with edge weights in $\{-n^{100k},\dots, n^{100k}\}$, determine whether the graph contains a $k$-clique whose edges sum to $0$.
\end{definition}

\begin{hypothesis}[\exactclique{k} hypothesis]
   In the word-RAM model with $O(\log n)$ bit words, there is no $O(n^{k-\epsilon})$ algorithm for $\exactclique{k}$, for any $\epsilon > 0$.
\end{hypothesis}

\section{Lower Bounds for 4-Cycle Enumeration}
In this section we prove \autoref{thm:four_cycle_lb} and show that the \fourcycle{} algorithm of Chiba and Nishizeki \cite{chiba1985arboricity} is tight. 

\FourCycleLB*

% To prove \autoref{thm:four_cycle_lb},  we use a construction by Brown \cite{brown1966graphs} to control the maximum degree of our graph, in combination with a technique introduced by Jin and Xu \cite{jin2023removing} regarding the hardness of the All-Edges Sparse Triangle problem, defined as follows.

To prove \autoref{thm:four_cycle_lb}, we use a result by Jin and Xu \cite{jin2023removing} regarding the hardness of determining for every edge of a graph whether or not it is part of a triangle. This problem, formally known as the All-Edge Sparse Triangle problem, is defined as follows.

\begin{definition}[All-Edge Sparse Triangle]
    Given a graph $G = (V, E)$, determine for every edge whether it is contained in a triangle.
\end{definition}

In their work, Jin and Xu show that the All-Edge Sparse Triangle problem is hard on graphs with few $k$-cycles for any integer $k\geq 3$. More precisely, we will use the following technical parameterized lemma from \cite{jin2023removing}.

\begin{lemma}[Lemma 6.1 in \cite{jin2023removing}]\label{lm:triangletofourcycle}
    Fix any constant $\sigma \in (0, 0.5)$, and any integer $k \geq 3$. Under the \threesum{} hypothesis, it requires $\Omega(n^{2-o(1)})$ time to solve $n^{3\sigma}$ instances of All-Edges Sparse Triangle on tripartite graphs with $O(n^{1-\sigma})$ vertices and maximum degree $O(n^{0.5-\sigma}),$ such that the total number of cycles of length at most $k$ over all instances is $O(n^{k/2 - (k-3)\sigma)})$.
\end{lemma}

In order to prove hardness for listing 4-cycles, we will begin with instances of All-Edge Sparse Triangle and construct a new graph where every triangle in the original graph corresponds to a 4-cycle in the new graph. To avoid having 4-cycles in the new graph that did not originate from a triangle, we make use of 4-cycle free graphs given by the following result by Brown \cite{brown1966graphs} and Erd\H{o}s et. al \cite{erdos1962problem}.

\begin{theorem}[\cite{brown1966graphs, erdos1962problem}]\label{thm:4-cycle-free}
    For every large $n$, there exists a 4-cycle free graph $G = (V, E)$ such that $|V| = \Theta(n)$, $|E| = \Theta(n^{3/2})$ with maximum degree $\Theta(n^{1/2})$. 
\end{theorem}

We can now prove \autoref{thm:four_cycle_lb}.

\begin{proof}[Proof of \autoref{thm:four_cycle_lb}]
    Let $\Delta = O(m^{1/3-\mu})$ be the maximum degree of the input graph $G$, for some $\mu \in (0, 1/3)$. It suffices to show that there is no \fourcycle{} algorithm with $O(n^{1-\delta} \Delta^2)$ pre-processing time and $n^{o(1)}$ delay for some $\delta > 0.$ Since $n\cdot \Delta \geq m$ and $\Delta \geq \alpha,$ this gives our desired bound.

    %Set $\sigma = \frac{1}{3}$ in \autoref{lm:triangletofourcycle}. We have that under the \threesum{} hypothesis, $\Omega(n^{2-o(1)})$ time is required to solve $n$ instances of \trianglelist{}  on tripartite graphs with $O(n^{2/3})$ vertices and maximum degree $O(n^{1/6})$, such that the total number of triangles is $O(n^{3/2})$ and the total number of 4-cycles is $O(n^{5/3})$. 

    Pick $\sigma = \frac{1-3\mu}{2-4\mu}$. It is easy to verify that $\sigma \in (0, 0.5).$ Consider the $n^{3\sigma}$ instances of All-Edges Sparse triangle as given in \autoref{lm:triangletofourcycle}. 
    
    Note that the resulting graph in each instance satisfies that max degree $\Delta = O(n^{0.5 - \sigma})$, number of vertices $n_0 = O(n^{1-\sigma})$, and number of edges $m_0 = O(n_0 \cdot d) = O(n^{1.5-2\sigma})$. To guarantee a lower bound of the number of vertices, edges and maximum degree, pad each instance with $n^\sigma$ ``dummy'' copies of a 4-cycle free graph as given in \autoref{thm:4-cycle-free} with $\Tilde{n} = n^{1-2\sigma}$. More formally, for every instance of All-Edge Sparse Triangle, construct a 4-cycle free graph with $\Tilde{n} = n^{1-2\sigma}$ vertices, $\Theta(\Tilde{n}^{3/2})$ edges and maximum degree $\Theta(\Tilde{n}^{1/2})$, and take the union of $n^\sigma$ copies of this graph and the All-Edge Sparse Triangle instance.
    
    Each such 4-cycle free graph we construct has $\Tilde{n}=\Theta(n^{1-2\sigma})$ vertices, $\Tilde{m}=\Theta(n^{1.5-3\sigma})$ edges and max-degree $\Theta(n^{0.5 - \sigma})$. Moreover, since any two triangles in a 4-cycle free graph must be edge-disjoint, each such graph has at most $O(\Tilde{m})=O(n^{1.5-3\sigma})$ triangles. 
    
    After taking its union with  $n^\sigma$ copies of the 4-cycle free graph, each resulting instance of All-Edge Sparse Triangle has $\Theta(n^{1-\sigma})$ vertices, $\Theta(n^{1.5-2\sigma})$ edges, and max-degree $\Theta(n^{0.5-\sigma})$. Each overall instance now satisfies that the maximum degree $\Delta = \Theta(n^{0.5 - \sigma})=\Theta(m^{1/3-\mu})$, and has at most $O(n^{1.5-2\sigma})$ triangles. In total, the $n^{3\sigma}$ instances of All-Edge Sparse Triangle put together have $O(n^{1.5+\sigma})$ triangles.
    
    Let $G = (V,E)$ be a graph in one such padded All-Edge Sparse Triangle instance. We can assume that $G$ is tripartite by using the color coding technique of Alon, Yuster and Zwick \cite{colorcoding} as follows. Assign each vertex in $V$ one of three colors uniformly and remove all edges whose two endpoints have the same color. The resulting graph is tripartite and every triangle in the original graph is maintained with constant probability. Therefore, if we are able to list all triangles in the colored version of $G$, then we can repeat the process of coloring and listing $\Tilde{O}(1)$ times and list every triangle in $G$ with high probability.
    
    Denote the tripartite vertex set of $G$ by $V = A\sqcup B\sqcup C$. Construct the 4-partite graph $G' = (V', E')$ as follows. Take $V'$ to be $A\cup B\cup C\cup C'$ where $C'$ is a copy of $C$. Take $E' = E(A,B)\cup E(B,C)\cup E(C',A)$ where $E(A,B),E(B,C)$ are the edges in $G$ between $A,B$ and $B,C$ respectively, and $E(C',A)$ is a copy of the edges in $G$ between $C$ and $A$. Finally, add a matching between $C$ and $C'$, where we match a vertex in $C$ to its copy in $C'$.

    Every triangle $a,b,c$ in the graph $G$ corresponds to the 4-cycle $a,b,c,c'$ in $G'$. 
    Therefore, listing all 4-cycles in $G'$ solves \trianglelist{} on $G$ and thus solves All-Edge Sparse Triangle. Note that every 4-cycle in $G'$ originated in either a triangle or a 4-cycle in $G$. Recall that there are at most $O(n^{1.5} + n^{1.5 + \sigma}) = O(n^{1.5+\sigma})$ triangles and $O(n^{2-\sigma})$ 4-cycles in $G$. 

    % Assume for contradiction there exists a \fourcycle{} algorithm with $O(N^{1-\delta}D^2)$ preprocessing time and $O(N^{o(1)})$ delay in a graph with $N$ vertices and maximum degree $D$. The graph $G'$ has $N = n^{2/3}$ vertices and maximum degree $D = n^{1/6}$, with $t$ 4-cycles. Therefore, the time required to list all 4-cycles in $G'$ is bounded by
    % \[
    % O\left(n^{\frac{2}{3}\cdot(1-\delta)}\cdot n^{\frac{1}{6}\cdot 2}+  n^{o(1)} \cdot t\right)=
    % O\left(n^{1-\frac{2}{3}\delta} + n^{o(1)}\cdot t\right).
    % \]

    Assume for contradiction there exists a \fourcycle{} algorithm with $O(N^{1-\delta}D^2)$ preprocessing time and $O(N^{o(1)})$ delay in a graph with $N$ vertices and maximum degree $D$. The graph $G'$ has $N = n^{1-\sigma}$ vertices and maximum degree $D = n^{0.5-\sigma}$, with $t$ 4-cycles. Therefore, the time required to list all 4-cycles in $G'$ is bounded by
    \[
    O\left(n^{(1-\sigma)(1-\delta)}\cdot n^{1 - 2\sigma}+  n^{o(1)} \cdot t\right)=
    O\left(n^{2 - \delta(1-\sigma) - 3\sigma} + n^{o(1)}\cdot t\right).
    \]

    Summing over all $n^{3\sigma}$ instances of $G'$, the total number of 4-cycles we will have to list is equal to the total number of triangles and 4-cycles in the original instances of the \trianglelist{} problems and so can be bounded by $O(n^{2-\sigma})$.
    
    Therefore, one can solve all $n^{3\sigma}$ instances of \trianglelist{} (and hence also all instances of All-Edges Sparse Triangle) in 
    \[O(n^{3\sigma} \cdot n^{2 - \delta(1-\sigma) - 3\sigma} + n^{1.5 + \sigma} \cdot n^{o(1)} + n^{2-\sigma} \cdot n^{o(1)}) = O(n^{2-\delta(1-\sigma)} + n^{1.5 + \sigma + o(1)} + n^{2 - \sigma + o(1)})
    \]
    time. Note that since $\delta(1-\sigma) \in (0.5\delta, \delta)$ and $\sigma \in (0, 0.5)$, this runtime is in fact sub-quadratic, thereby leading to a contradiction. 
    
    We conclude that under the \threesum{} hypothesis, no algorithm with $O(m^{1-\varepsilon}\alpha)$ preprocessing time and $m^{o(1)}$ delay can list all 4-cycles in the graph. Therefore, Chiba and Nishizeki's \cite{chiba1985arboricity} 4-cycle listing algorithm (\autoref{thm:CN-four-cycle}) is in fact tight.

\end{proof}

\section{Lower Bounds for \texorpdfstring{$k$}{k}-clique Enumeration}
In this section we prove \autoref{thm:k_clique_lb} and show that the \cliquelist{k} algorithm proposed in \cite{chiba1985arboricity} is essentially tight. We build on the result of Vassilevska W. and Xu \cite{williams2020monochromatic}, who show that triangle listing in a graph with $m$ edges and arboricity $\alpha$ requires $\Omega(m\alpha^{1-o(1)})$ time. 
We extend this result to the problem of \cliquelist{k} using a generalized approach to edge-weight hashing introduced in \cite{dalirrooyfard2023listing}. Our proof reduces the problem of finding a zero-weight $k$-clique to the problem of listing $k$-cliques, and relies on the \exactclique{k} hypothesis.

\KCliqueLB*

\begin{proof}
    Similar to the proof of \autoref{thm:four_cycle_lb}, we show that \cliquelist{k} in a graph with maximum degree $\Delta$ requires $\Omega(n^{1-O(\delta)}\cdot \Delta^{k-1})$ time. Since $n\cdot \Delta \geq m$ and $\Delta \geq \alpha$ this gives us our desired bound. 

    As we showed in the case of triangle listing, we can use a standard color coding technique \cite{colorcoding} to assume that the graph is $k$-partite. Denote the vertex set of $G$ by $V = V_1 \sqcup V_2\sqcup \ldots \sqcup V_k$, each of size $n$. 

    We begin by using an edge hashing scheme introduced in \cite{dalirrooyfard2023listing} to modify the edge weights of the graph. This will ensure that the edges of a clique behave as if their weights were assigned uniformly and independently at random.
    
    Denote the original weight of an edge $uv\in E$ by $w(u,v)$ and assume all $w(u,v)\in \FF_p$ for some prime $p$. Sample a value $x\sim \FF_p$ uniformly at random. For each $i\in [k]$ and $v\in V_i$  sample $k-1$ values $y_{v,1}, y_{v,2}, \ldots, y_{v,i-1},y_{v,i+1},\ldots, y_{v,k}$  such that $\sum_{j\neq i}y_{v,j} = 0$. Define a new weight function on the edges as follows. For any $1\leq i < j \leq k$ and any $v_i\in V_i, v_j \in V_j$ define 
    \[
    w'(v_i, v_j) = x\cdot w(v_i, v_j) + y_{v_i, j} + y_{v_j, i}.
    \]
    
    Informally, these weights act like uniformly random, independently sampled random variables. Formally, the modified weight function $w'$ has the following property:
    \begin{lemma} [Lemma 4.2 and Corollary 4.3, in the full version of \cite{dalirrooyfard2023listing}]
        If $(v_1, \ldots, v_k)$ is not an exact-$k$-clique w.r.t $w$, and $(u_1, \ldots, u_k)$ shares exactly $c$ vertices indexed by $S$ with $(v_1, \ldots, v_k)$, where $1\leq c \leq k$, then the random variables 
        \[
        \left\{w'(u_i, u_j) - w'(v_i, v_j)\right\}_{1 \leq i < j \leq k, i\notin S \lor j\notin S}
        \]
        are independent. Furthermore, for any $k$-clique on vertices $(v_i)_{i\in T}$ that shares exactly $c$ vertices indexed by $S$ with $(u_1, \ldots, u_k)$, the random variables 
        \[
        \left\{w'(u_i, u_j) - w'(v_i, v_j)\right\}_{i,j\in T, i < j, i\notin S \lor j\notin S}
        \]
        are independent.
    \end{lemma}

    Denote by $L$ the range of values of the modified weight function and let $s$ be a parameter to be set at a later time. Partition $L$ into $s$ continuous intervals of length $O(L/s)$, $L_1, \ldots, L_s$. For any ${k \choose 2}$ tuple of intervals, $(L_{i_1}, L_{i_2}, \ldots, L_{i_{k \choose 2}})$, consider the graph $G_{i_1, \ldots, i_{k\choose 2}}$ obtained by keeping the edges between each pair of sets $V_i, V_j$ whose weights fall into the appropriate range (e.g. keeping the edges between $V_1$ and $V_2$ whose weights fall in the range $L_{i_1}$, keeping the edges between $V_1$ and $V_3$ whose weights fall in the range $L_{i_2}$ and so on). Note that to find a zero-$k$-clique we only need to consider ${k \choose 2}$ tuples such that $0 \in L_{i_1} + \ldots + L_{i_{k \choose 2}}$. Since the sumset $L_{i_1} + \ldots + L_{i_{k\choose 2} - 1}$ is of size $\leq ({k \choose 2} - 1) \cdot s = O(s)$, there are $O(1)$ possibilities for the choice of $L_{i_{k \choose 2}}$. Therefore, it suffices to consider $O(s^{{k \choose 2} - 1})$ tuples of intervals to find any zero-$k$-clique instance.
    
    For each such graph graph $G_{i_1, \ldots, i_{k\choose 2}}$, with high probability, the degree of each vertex is $O(n/s)$. Furthermore, with high probability the number of $k$-cliques in each instance is $O\left(\frac{n^k}{s^{k \choose 2}}\right)$, since for every set of $k$ nodes, each edge between them remains with probability $\frac{1}{s}$.
    
    We can now use this partition to solve \exactclique{k}. For each instance $G_{i_1, \ldots, i_{k\choose 2}}$ such that $0\in L_{i_1} + \ldots + L_{i_{k\choose 2}}$, run $k$-clique listing. For each $k$-clique listed by the algorithm, check if it is a Zero-$k$-clique. Clearly this algorithm will find a Zero-$k$-clique if one exists. We are left to calculate its running time.
    
    Suppose that we have an algorithm for $k$-clique listing that runs in time $O(n^{1-\epsilon}\cdot \Delta^{k-1})$. We run $k$-clique listing $s^{{k \choose 2} - 1}$ times, where each instance has degree $\Delta \leq  \frac{n}{s}$. The total running time of this is
    \[
    s^{{k \choose 2} - 1}\cdot n^{1 - \epsilon}\cdot \left( \frac{n}{s} \right)^{k-1} = 
    n^{k - \epsilon}\cdot s^{{k \choose 2} - k}.
    \]
    
    Since each instance generates at most $O\left(\frac{n^k}{s^{k \choose 2}}\right)$ $k$-cliques, checking whether any $k$-clique is a Zero-$k$-clique takes
    \[
    s^{{k \choose 2} - 1}\cdot \frac{n^k}{s^{k \choose 2}} = \frac{n^k}{s}.
    \]
    Setting $s = n^{\frac{\epsilon}{{k \choose 2} - k  + 1}} = n^{\frac{2 \epsilon}{k^2 - 3k + 2}}$ gives us a total running time of $O(n^{k- \frac{2 \epsilon}{k^2 - 3k + 2}})$. 
    
    Therefore, an $n^{1-\epsilon}\Delta^{k-1}$ time algorithm for $k$-clique listing on graphs with $\alpha \leq \Delta \leq n^{1 - \Theta(\epsilon)}$ gives an $o(n^k)$ time algorithm for Zero-$k$-clique, thus contradicting the \exactclique{k} hypothesis. We conclude that under the \exactclique{k} hypothesis, if $\alpha = n^{1-\delta}$, no $O(m^{1-O(\delta)}\alpha^{k-2})$ time algorithm can list all $k$-cliques in the graph.

\end{proof}

\ifnum\anon=0
\section*{Acknowledgements}
We would like to thank Virginia Vassilevska Williams and Ryan Williams for their enthusiasm and guidance. We would like to thank Yinzhan Xu for many helpful discussions.
\fi

\bibliography{ref}

\newcommand{\etalchar}[1]{$^{#1}$}
\begin{thebibliography}{DMVWX24}

\bibitem[ABDN18]{abboud2018more}
Amir Abboud, Karl Bringmann, Holger Dell, and Jesper Nederlof.
\newblock More consequences of falsifying seth and the orthogonal vectors
  conjecture.
\newblock In {\em Proceedings of the 50th Annual ACM SIGACT Symposium on Theory
  of Computing (STOC)}, pages 253--266, 2018.

\bibitem[ABF23]{abboud2023stronger}
Amir Abboud, Karl Bringmann, and Nick Fischer.
\newblock Stronger 3-sum lower bounds for approximate distance oracles via
  additive combinatorics.
\newblock In {\em Proceedings of the 55th Annual ACM Symposium on Theory of
  Computing}, pages 391--404, 2023.

\bibitem[AKLS22]{abboud2022listing}
Amir Abboud, Seri Khoury, Oree Leibowitz, and Ron Safier.
\newblock Listing 4-cycles.
\newblock {\em arXiv preprint arXiv:2211.10022}, 2022.

\bibitem[AVW14]{AbboudWW14}
Amir Abboud, {Virginia {Vassilevska Williams}}, and Oren Weimann.
\newblock Consequences of faster alignment of sequences.
\newblock In {\em Proceedings of the 41st International Colloquium on Automata,
  Languages, and Programming (ICALP)}, pages 39--51, 2014.

\bibitem[AYZ95]{colorcoding}
Noga Alon, Raphael Yuster, and Uri Zwick.
\newblock Color-coding.
\newblock {\em J. ACM}, 42(4):844–856, jul 1995.

\bibitem[AYZ97]{findingcycles}
Noga Alon, Raphael Yuster, and Uri Zwick.
\newblock Finding and counting given length cycles.
\newblock {\em Algorithmica}, 17(3):209--223, 1997.

\bibitem[BCM22]{BringmannCM22}
Karl Bringmann, Nofar Carmeli, and Stefan Mengel.
\newblock Tight fine-grained bounds for direct access on join queries.
\newblock In {\em Proceedings of the 41st ACM SIGMOD-SIGACT-SIGAI Symposium on
  Principles of Database Systems (PODS)}, pages 427--436, 2022.

\bibitem[BDT16]{BackursDT16}
Arturs Backurs, Nishanth Dikkala, and Christos Tzamos.
\newblock Tight hardness results for maximum weight rectangles.
\newblock In {\em Proceedings of the 43rd International Colloquium on Automata,
  Languages, and Programming (ICALP)}, pages 81:1--81:13, 2016.

\bibitem[BGMW20]{BringmannGMW20}
Karl Bringmann, Pawel Gawrychowski, Shay Mozes, and Oren Weimann.
\newblock Tree edit distance cannot be computed in strongly subcubic time
  (unless {APSP} can).
\newblock {\em {ACM} Trans. Algorithms}, 16(4):48:1--48:22, 2020.

\bibitem[BPVZ14]{bjorklund2014listing}
Andreas Bj{\"o}rklund, Rasmus Pagh, Virginia {Vassilevska Williams}, and Uri
  Zwick.
\newblock Listing triangles.
\newblock In {\em Proceedings of the 41st International Colloquium on Automata,
  Languages, and Programming (ICALP)}, pages 223--234, 2014.

\bibitem[Bro66]{brown1966graphs}
William~G Brown.
\newblock On graphs that do not contain a thomsen graph.
\newblock {\em Canadian Mathematical Bulletin}, 9(3):281--285, 1966.

\bibitem[BT17]{BackursT17}
Arturs Backurs and Christos Tzamos.
\newblock Improving viterbi is hard: Better runtimes imply faster clique
  algorithms.
\newblock In {\em Proceedings of the 34th International Conference on Machine
  Learning (ICML)}, pages 311--321, 2017.

\bibitem[CC11]{ChuC11}
Shumo Chu and James Cheng.
\newblock Triangle listing in massive networks and its applications.
\newblock In {\em Proceedings of the 17th {ACM} {SIGKDD} International
  Conference on Knowledge Discovery and Data Mining (KDD)}, pages 672--680,
  2011.

\bibitem[CN85]{chiba1985arboricity}
Norishige Chiba and Takao Nishizeki.
\newblock Arboricity and subgraph listing algorithms.
\newblock {\em SIAM Journal on computing}, 14(1):210--223, 1985.

\bibitem[DBS18]{danisch2018listing}
Maximilien Danisch, Oana Balalau, and Mauro Sozio.
\newblock Listing k-cliques in sparse real-world graphs.
\newblock In {\em Proceedings of the 2018 World Wide Web Conference}, pages
  589--598, 2018.

\bibitem[DGP09]{dourisboure2009extraction}
Yon Dourisboure, Filippo Geraci, and Marco Pellegrini.
\newblock Extraction and classification of dense implicit communities in the
  web graph.
\newblock {\em ACM Transactions on the Web (TWEB)}, 3(2):1--36, 2009.

\bibitem[DMVWX24]{dalirrooyfard2023listing}
Mina Dalirrooyfard, Surya Mathialagan, Virginia Vassilevska~Williams, and
  Yinzhan Xu.
\newblock Towards optimal output-sensitive clique listing or: Listing cliques
  from smaller cliques.
\newblock In {\em Proceedings of the 56th Annual ACM Symposium on Theory of
  Computing}, STOC 2024, page 923–934, New York, NY, USA, 2024. Association
  for Computing Machinery.

\bibitem[EGMT17]{eppstein2017}
David Eppstein, Michael~T. Goodrich, Michael Mitzenmacher, and Manuel~R.
  Torres.
\newblock 2-3 cuckoo filters for faster triangle listing and set intersection.
\newblock In {\em Proceedings of the 36th ACM SIGMOD-SIGACT-SIGAI Symposium on
  Principles of Database Systems}, PODS '17, page 247–260, New York, NY, USA,
  2017. Association for Computing Machinery.

\bibitem[ERS62]{erdos1962problem}
Paul Erdos, Alfr{\'e}d R{\'e}nyi, and VT~S{\'o}s.
\newblock On a problem in the theory of graphs.
\newblock {\em Publ. Math. Inst. Hungar. Acad. Sci}, 7:215--235, 1962.

\bibitem[FNBB06]{fratkin2006motifcut}
Eugene Fratkin, Brian~T Naughton, Douglas~L Brutlag, and Serafim Batzoglou.
\newblock Motifcut: regulatory motifs finding with maximum density subgraphs.
\newblock {\em Bioinformatics}, 22(14):e150--e157, 2006.

\bibitem[JX23]{jin2023removing}
Ce~Jin and Yinzhan Xu.
\newblock Removing additive structure in 3sum-based reductions.
\newblock In {\em Proceedings of the 55th Annual ACM Symposium on Theory of
  Computing}, pages 405--418, 2023.

\bibitem[KPP15]{kopelowitz2015dynamic}
Tsvi Kopelowitz, Seth Pettie, and Ely Porat.
\newblock Dynamic set intersection.
\newblock In {\em Algorithms and Data Structures: 14th International Symposium,
  WADS 2015, Victoria, BC, Canada, August 5-7, 2015. Proceedings 14}, pages
  470--481. Springer, 2015.

\bibitem[KPP16]{kopelowitz2016higher}
Tsvi Kopelowitz, Seth Pettie, and Ely Porat.
\newblock Higher lower bounds from the 3sum conjecture.
\newblock In {\em Proceedings of the twenty-seventh annual ACM-SIAM symposium
  on Discrete algorithms}, pages 1272--1287. SIAM, 2016.

\bibitem[Lat08]{trilistlatapy}
Matthieu Latapy.
\newblock Main-memory triangle computations for very large (sparse (power-law))
  graphs.
\newblock {\em Theor. Comput. Sci.}, 407(1):458--473, 2008.

\bibitem[LVW18]{LincolnWW18}
Andrea Lincoln, Virginia {Vassilevska Williams}, and R.~Ryan Williams.
\newblock Tight hardness for shortest cycles and paths in sparse graphs.
\newblock In {\em Proceedings of the 29th Annual {ACM-SIAM} Symposium on
  Discrete Algorithms (SODA)}, pages 1236--1252, 2018.

\bibitem[PSV17]{count5via3}
Ali Pinar, C.~Seshadhri, and Vaidyanathan Vishal.
\newblock Escape: Efficiently counting all 5-vertex subgraphs.
\newblock In {\em Proceedings of the 26th International Conference on World
  Wide Web (WWW)}, pages 1431--1440, 2017.

\bibitem[QCQ{\etalchar{+}}18]{qiu2018real}
Xiafei Qiu, Wubin Cen, Zhengping Qian, You Peng, Ying Zhang, Xuemin Lin, and
  Jingren Zhou.
\newblock Real-time constrained cycle detection in large dynamic graphs.
\newblock {\em Proceedings of the VLDB Endowment}, 11(12):1876--1888, 2018.

\bibitem[SDS21]{shi2021parallel}
Jessica Shi, Laxman Dhulipala, and Julian Shun.
\newblock Parallel clique counting and peeling algorithms.
\newblock In {\em SIAM Conference on Applied and Computational Discrete
  Algorithms (ACDA21)}, pages 135--146. SIAM, 2021.

\bibitem[SSPUVc15]{listingcliquesnucleus}
Ahmet~Erdem Sariy\"{u}ce, C.~Seshadhri, Ali Pinar, and \"{U}mit
  V.~\c{C}ataly\"{u}rek.
\newblock Finding the hierarchy of dense subgraphs using nucleus
  decompositions.
\newblock In {\em Proceedings of the International Conference on the World Wide
  Web (WWW)}, pages 927--937, 2015.

\bibitem[ST15]{ShunT15}
Julian Shun and Kanat Tangwongsan.
\newblock Multicore triangle computations without tuning.
\newblock In {\em Proceedings of the 31st {IEEE} International Conference on
  Data Engineering (ICDE)}, pages 149--160, 2015.

\bibitem[SW05]{SchankW05}
Thomas Schank and Dorothea Wagner.
\newblock Finding, counting and listing all triangles in large graphs, an
  experimental study.
\newblock In {\em Proceedings of the 4th International Conference on
  Experimental and Efficient Algorithms (WEA)}, pages 606--609, 2005.

\bibitem[Tso15]{listingcliquesdensest}
Charalampos Tsourakakis.
\newblock The k-clique densest subgraph problem.
\newblock In {\em Proceedings of the International Conference on the World Wide
  Web (WWW)}, pages 1122--1132, 2015.

\bibitem[{Vas}18]{vsurvey}
Virginia {Vassilevska Williams}.
\newblock On some fine-grained questions in algorithms and complexity.
\newblock In {\em Proceedings of the ICM}, volume~3, pages 3431--3472. World
  Scientific, 2018.

\bibitem[VWX20]{williams2020monochromatic}
Virginia Vassilevska~Williams and Yinzhan Xu.
\newblock Monochromatic triangles, triangle listing and apsp.
\newblock In {\em 2020 IEEE 61st Annual Symposium on Foundations of Computer
  Science (FOCS)}, pages 786--797. IEEE, 2020.

\bibitem[YBLG17]{localclustering}
Hao Yin, Austin~R Benson, Jure Leskovec, and David~F Gleich.
\newblock Local higher-order graph clustering.
\newblock In {\em Proceedings of the 23rd ACM SIGKDD international conference
  on knowledge discovery and data mining}, pages 555--564, 2017.

\end{thebibliography}

\end{document}